\documentclass[runningheads]{llncs}

\usepackage{latexsym,amssymb,amsmath}
\usepackage{amsfonts,amscd,color,epsfig}
\usepackage{pst-all}
\usepackage{float}

\floatstyle{ruled}
\newfloat{protocol}{tph}{lop}
\floatname{protocol}{Protocol} 
\providecommand{\abs}[1]{\lvert#1\rvert}

\begin{document}

\title{Information-Theoretically Secure Voting\\
Without an Honest Majority}
\titlerunning{Information-Theoretically Secure Voting}

\author{Anne Broadbent  \and   Alain Tapp}
% \vspace{15pt}

\institute{D\'epartement d'informatique et de recherche op\'erationnelle\\
 Universit\'e de Montr\'eal, C.P.~6128, Succ.\ Centre-Ville\\
 Montr\'eal (QC), H3C 3J7~~\textsc{Canada}\\
\url{\{broadbea, tappa\}@iro.umontreal.ca}}

%AB author names now appear

%\normalsize\sl D\'epartement d'informatique et de recherche op\'erationnelle\\[-0.1cm]
%\normalsize\sl Universit\'e de Montr\'eal, C.P.~6128, Succ.\ Centre-Ville\\[-0.1cm]
%\normalsize\sl Montr\'eal (QC), H3C 3J7~~\textsc{Canada}\\[.2cm]
%\url{{broadbea, tappa}@iro.umontreal.ca} }
%
%\author{Anne Broadbent \and Alain Tapp}

%\institute{Universit\'{e} de Montr\'{e}al \\
%D\'{e}partement d'informatique et de recherche op\'{e}rationnelle \\
%C.P.\ 6128, Succ.\ Centre-Ville, Montr\'{e}al (Qu\'ebec), H3C 3J7 \textsc{Canada} \\
%\email{\{broadbea,tappa\}\textnormal{@}iro.umontreal.ca}}

%\date{June 13, 2007}

\maketitle

\begin{abstract}
We present three voting protocols with unconditional privacy and
information-theoretic correctness, without assuming any bound on the
number of corrupt voters or voting authorities. All protocols have
polynomial complexity and require private channels and a
simultaneous broadcast channel. Our first protocol is a basic voting
scheme which allows voters to interact in order to compute the
tally. Privacy of the ballot is unconditional, but any voter can
cause the protocol to fail, in which case information about the
tally may nevertheless transpire. Our second protocol introduces
voting authorities which allow the implementation of the first
protocol, while reducing the interaction and limiting it to be only
between voters and authorities and among the authorities themselves.
The simultaneous broadcast is also limited to the authorities.
 As long as a single authority is honest,
the privacy is unconditional, however, a single corrupt authority or
a single corrupt voter can cause the protocol to fail. Our final
protocol provides a safeguard against corrupt voters by enabling a
verification technique to allow the authorities to revoke
incorrect votes.
We also discuss the implementation of a simultaneous broadcast
channel with the use of temporary computational assumptions,
yielding versions of our protocols achieving  everlasting security.  \looseness-1\\
%\vspace{.5cm}
\smallskip

%\\[1ex]
\end{abstract}

\section{Introduction}

Multiparty secure computation enables
a group of~$n$ participants  to collaborate in order to compute a
function on their private inputs. Assuming that private random keys
are shared between each pair of participants,  every function can be
securely computed
if and only if less than~$n/3$
participants are corrupt; this fundamental result is due to  David
Chaum, Claude Cr\'epeau and Ivan Damg\a rard~\cite{CCD88} and to
Michael Ben-Or, Shafi Goldwasser and Avi Wigderson~\cite{BGW88}.
When a broadcast channel is available, the results of Tal Rabin and
Michael \mbox{Ben-Or}~\cite{RB89} tell us that this proportion can
be improved to~$n/2$.

Among all functions that can be computed with these general-purpose
protocols, perhaps the one that has the most obvious application is
voting. If we have a guarantee on the proportion of honest
participants, a secure voting protocol based only on pairwise
private channels can be implemented %AB new
(if, in addition to this, we have a broadcast channel, then we can
tolerate more cheaters). %
Here, we are interested in the case where no such guarantee is
available. The first protocol for voting that is
information-theoretically secure even in a presence of a majority of
dishonest participants was presented in~\cite{BroadbentTapp}. Along
with the use of private communication, the protocol uses a
simultaneous broadcast channel. In this extended abstract, we first
give a new presentation of the original protocol, followed by two
protocols which present significant improvements on the original
one. Although our initial motivation  was of theoretical nature, we
believe that this work may lead to interesting practical
applications.

All three protocols are obtained from two simple yet powerful
observations. First, if the dinning cryptographer's
protocol~\cite{Chaum88} is used to compute the parity function and
is implemented with a simultaneous broadcast channel, then it is
perfect. The second observation is that if a string of $n$ bits is
shared among $n$ participants is such that the parity of the $n$
bits is random (and unknown), then it is impossible for any strict
subset of participants to locally derandomize this parity.

In our first protocol, we assume that each pair of voters is
connected by a private authentic channel. In our second and third
protocols, we relax this assumption by  introducing \emph{voting
authorities}. The assumption then becomes that there are private and
authentic channels only between voters and the authorities and among
the authorities themselves.

All three protocols  require a  simultaneous broadcast
channel~\cite{CGMW85,HD05}, which, for our purpose, is  a collection
of broadcast channels where the input of one participant cannot
depend on the input of any other participant. This could be achieved
if all participants {\em simultaneously} performed a broadcast. In
the context of our second and third protocols, a simultaneous
broadcast among the authorities is sufficient.

It is not uncommon in multiparty computation to allow additional
resources, even if these resources cannot be implemented with the
threshold on the honest participants (the results of~\cite{RB89}
which combine a broadcast channel with~$n/2$ honest participants
being the most obvious example). Our work suggests that a
simultaneous broadcast channel is an interesting primitive to study
in this context. Furthermore, given a resource to implement bit
commitment, we can implement a simultaneous broadcast: all
participants commit to their values, and then all participants open
these values. Since bit commitment can be implemented based on the
laws of relativity~\cite{Kent99} (or more precisely, based on the
postulate that information cannot travel faster then the speed of light),
we conclude that simultaneous broadcast can also be achieved in this
model. It may also be possible to directly implement a simultaneous
broadcast using the laws of relativity.

Since a simultaneous broadcast channel can be achieved using bit
commitment, which itself can be implemented with computational
assumptions, we can replace in all our protocols the use of a
simultaneous broadcast channel with temporary computational
assumptions. Our protocols then provide \emph{everlasting} security:
as long as the computational assumptions are not broken
\emph{during} the execution of the protocol (more precisely, during
the simulation of the simultaneous broadcast), the security of the
protocols is perfect. Note that the privacy of individual votes
remains perfect even if these computational assumptions are broken
during the protocol.

\subsection{Common Features to All Protocols}

Our voting protocols involve $n$~voters, each casting a ballot for a
single choice among $m$~candidates. The goal of the protocols is to
faithfully count the number of ballots in favour of each candidate
in such a way that voter's ballots remain private, honest ballots
are counted and dishonest voters cannot influence the vote any more
than by honestly voting. The protocols we present are based on a
technique presented in \cite{BroadbentTapp}. The first protocol
involves only the voters but the last two involve $r$ voting
authorities. In all protocols, dishonest participants can make the
protocol fail (in our last protocol, only dishonest authorities can
achieve this). All three protocols use probabilistic techniques to
correctly evaluate the tally for each candidate. For this reason,
the protocols are only correct with probability $1-2^{-\Omega(s)}$,
with~$s$ being a chosen security parameter.

We present our protocols in the regular setup  where each voter
casts a ballot with a choice for a single candidate. Our protocols
can easily be adapted to allow any number of voices per ballot
(allowing, for instance, each voter to either choose two candidates,
or to vote twice for the same candidate). We can also add a dummy
candidate to allow voters to honestly cancel their ballots.

\subsection{Summary of Results}

All three protocols are exclusively based on private authentic
channels and a simultaneous broadcast channel. In the first
protocol, no assumption is made on the number of honest voters and
in the last two, the only assumption is that at least one authority
is honest. Under these assumption, our protocols provide perfect
privacy and correctness. This was believed to be
impossible~\cite{Jeroen}. The major drawback is that any dishonest
participant can  make any protocol fail (except in our third
protocol, where only dishonest authorities can make the protocol
fail).

\textbf{Protocols~\ref{prot:vote2}} and \textbf{\ref{prot:vote3}}
make use of voting authorities. If we group the authorities
together, they act as a trusted third party, which means that
collectively they can violate privacy and correctness of the
protocol. However, taken individually, both privacy and correctness
are guaranteed as long as a \emph{single} authority is honest. This
suggests that in practice, authorities could be chosen to represent
different interest groups, with each voter needing to trust only a
single authority (note that it is not necessary for the voters to
trust the \emph{same} authority!).

It is common in multiparty computation to compare an implementation
of a functionality with its \emph{ideal} functionality. This ideal
functionality is represented as a black box, accepting private
inputs from each participant and privately communicating the
function evaluation on these private inputs back to each
participant. We now review the main features of each protocol.

\subsubsection{Basic Voting (section~\ref{sec:vote})}

\begin{itemize}
\item Only voters are involved in the protocol.
\item A coalition of dishonest voters can only learn through the protocol what they would learn in the
ideal functionality, and this even (and also) if the protocol fails.
\item A single dishonest voter can make the protocol fail.
\item If the protocol does not fail, then it is consistent with all ballots of the honest voters and
some assignment of ballots for the dishonest voters.
\item Dishonest voters cannot vote adaptively.
\end{itemize}

\subsubsection{Voting with Authorities (section~\ref{sec:votewithauthorities})}
\begin{itemize}
\item Voters and a small number of authorities are involved in the protocol.
\item Voters only interact with authorities.
\item If at least one authority is honest, a coalition of dishonest voters and authorities
can only learn
what they would learn in the
ideal functionality, and this even (and also) if the protocol fails.
\item A single dishonest voter or authority can make the protocol fail.
\item If at least one authority is honest and if the protocol does not fail, then
it is consistent with all ballots of the honest voters and some
assignment of ballots for the dishonest voters.
\item If at least one authority is honest, a coalition of dishonest voters and authorities cannot vote adaptively.
\end{itemize}

\subsubsection{Voting with Authorities and Verification (section~\ref{sec:votewithverification})}
\begin{itemize}
\item Voters and a small number of authorities are involved in the protocol.
\item Voters only interact with authorities.
\item If at least one authority is honest, a coalition of dishonest voters and authorities
can only learn
what they would learn in the ideal functionality, and this even (and
also) if the protocol fails.
\item No coalition of voters alone can make the protocol fail.
\item A single dishonest authority can make the protocol fail.
\item If at least one authority is honest and if the protocol does not fail, then
it is consistent with all ballots of the honest voters and some
assignment of ballots for the dishonest voters.
\item If at least one authority is honest, a coalition of dishonest voters and authorities cannot vote adaptively.
\item Dishonest voters voting inappropriately will have their ballot revoked.
\item A dishonest authority can choose to revoke the ballot of an honest voter.
\item When a ballot is revoked, all voters and authorities know
about it.
\end{itemize}

%%%%%%%%%%%%%%%%%%%%%%%%%%%%%%%%%%%%%%%%%%%%%%%%%%%%%%%%%%%%%%%%%%%%%%%%%%%%%%%%%%%%%%
\section{Basic Voting Protocol}
\label{sec:vote}

We present a protocol that allows $n$ voters to conduct an
$m$-candidate~vote. First, some notation: we say that participants
share a \emph{distributed bit with value $b$} if each participant
holds a bit and the parity (binary XOR) of all bits is~$b$. Within a
group of $n$ participants, we say that a voter \emph{constructs} a
distributed bit with value~$b$ if he chooses $b_i\in_R \{0,1\}$ such
that $\bigoplus_{i=1}^n b_i=b$ and  sends privately $b_i$ to
participant~$i$. The values $\{b_i\}$ ($i=1,\ldots n)$ are called
\emph{shares}.  For now, voters create distributed bits among
themselves. In sections~\ref{sec:votewithauthorities}
and~\ref{sec:votewithverification},  voters will create distributed
bits among authorities. Our basic protocol is given as
\textbf{Protocol~\ref{prot:vote1}}.

%%%%%%%%%%%%%%%%%%%%%%%%%%%%%%%%%%%%%%%%%%%%%%%%%%%%%%%%%%%%%%%%%%%%%%%%%%%%%%%%%%%%%%
\begin{protocol}[h]
\caption{Basic voting protocol} \label{prot:vote1}
{\bf Input:} $x_i \in \{1,\ldots,m\}$
%AB New
and security parameter $s$ \\
{\bf Output:} for $k=1$ to $m$, $y[k]= \abs{\{x_j \mid x_j=k \}}$

\vspace{4pt} \hrule \vspace{4pt}

\textbf{Phase A (cast)}

For each candidate $k=1$ to~$m$, \vspace{-.25cm}
\begin{enumerate}
\item
Each voter~$i$ sets the value of $n^2s$ bits~$p_{ijk}$ ($j =1,
\ldots ,n^2 s$) in the following way: if~$x_i \neq k$, then all bits
are $0$; otherwise, exactly $ns$ bits (a fraction $1/n$ of the
total) are randomly chosen such that $p_{ijk}=1$ and the rest such
that $p_{ijk}=0$.

\item
For each $j =1, \ldots ,n^2 s$, each voter $i$ constructs a
distributed bit  with value $p_{ijk}$. Let the shares of each
distributed bit be denoted $\{p_{ijk\ell}\}$ ($\ell = 1, \ldots n$)
\end{enumerate}

\textbf{Phase B (broadcast)}

For every $j$ and $k$, each voter $\ell$, computes the parity of all
received bits, $q_{jk\ell}=\bigoplus_{i=1}^n p_{ijk\ell}$. All bits
are then simultaneously broadcast.

\textbf{Phase C (tally)}

To compute the tally, $y[k]$, for each value $k=1, \ldots ,m$, each
voter sets: \mbox{$v[k]_j=\bigoplus_{\ell=1}^n q_{jk\ell}$},
\mbox{$\sigma[k] = \sum_{j=1}^{n^2 s} \frac{v[k]_j}{n^2s}$} and  if
there exists an integer~$v$ such that $\left|\sigma[k] - p_v \right|
< \frac{1}{2e^2n} \,$, where $p_v =
\frac{1}{2}\left(\frac{n-2}{n}\right)^v
\left(\left(\frac{n}{n-2}\right)^v -1\right)$, then $y[k]=v$\,.

If for any~$m$, no such value~$v$ exists, or if $\sum_{k=1}^m y[k]
\neq n$, the protocol fails.

\end{protocol}

The complexity of \textbf{Protocol~\ref{prot:vote1}} is as follows:
$n$ voters each create $mn^2s$ distributed bits, for a total of $n$
messages of size $mn^2s$. \textbf{Phase B} requires a single
simultaneous broadcast among $n$ participants, each sending a
message of size $m n^2 s$.

\begin{lemma}(Correctness)
\label{lem:correctness corrupt} If
\textbf{Protocol~\ref{prot:vote1}} does not fail, the result of the
vote is consistent with the vote of the honest voters and some
non-adaptive choice for the dishonest voters, except with
probability exponentially small in~$s$.
%AB: note: I don't like the position of except above, but don't have an easy solution.
\end{lemma}

\begin{proof}
Our protocol is presented in a way that minimizes the number of
messages sent by each voter; it is perhaps best understood
intuitively in its sequential version. From this point of view, the
following is repeated~$n^2 s$ times. For each candidate, voters
create a distributed bit. The value of the distributed bit is~1 with
probability~$1/n$ if this is the candidate the voter chooses  and
always 0 otherwise. All voters compute the XOR of all their shares
and the result will eventually be simultaneously broadcast. The
probability that the parity of the broadcast value is~1 is directly
proportional to the number of voters voting for the candidate. By
repeating this process with each candidate $n^2 s$ times, we can
gather enough statistics to compute the vote exactly with very high
probability.

The only place a voter can deviate from the protocol is by creating
distributed bits with an inappropriate ratio of 0 and 1 values. We
first note that if the corrupted voters actually transmit the
correct number of private bits in \textbf{phase~A} and broadcast the
correct number of bits in \textbf{phase~B}, then whatever they
actually send is consistent with some global ratio of even and odd
distributed bits.

The ratio of even and odd distributed bits, when XORed, will give
rise to some probability of an even or an odd bit in the
simultaneous broadcast. It is possible to randomize the parity but
not to derandomize it: the corrupt participants altogether can
increase the probability of an odd broadcast but not make it
smaller. Because votes for each candidate are added up for a
consistency check, either the corrupted voters make a consistent
number of votes or otherwise the protocol will fail. %AB new
The use of a simultaneous broadcast channel ensures that the voter's
inputs are independent of each other.

In the rest of the proof, we give a detailed analysis, using a
Chernoff-type argument that the result of the vote will be correct
with overwhelming probability.

We fix a value~$k$ and suppose that~$v$ voters have input~$x_i=k$.
Thus we need to show that  in \textbf{Protocol~\ref{prot:vote1}},
$y[k] = v$,  except with probability exponentially small in~$s$.

Let us look at \textbf{phase~C} of the protocol. Let~$p_v$ be the
probability that $v[k]_j = 1$. For  $v \leq n$, we have~$p_0 = 0$,
$p_1 = \frac{1}{n}$ and $p_{v+1} = p_v\left(1-\frac{1}{n}\right) +
(1-p_v)\frac{1}{n}$. Solving this recurrence, we get
\begin{equation} p_v =
\frac{1}{2}\left(\frac{n-2}{n}\right)^v
\left(\left(\frac{n}{n-2}\right)^v -1\right) \,. \end{equation}
Thus, the idea of \textbf{phase~C} is for the participants to
approximate~$p_v$ by computing \mbox{$\sigma[k] = \sum_{i=1}^{n^2s}
v[k]_j/{n^2s}$}.
 If the approximation
is within~$\frac{1}{2e^2n}$ of~$p_v$, then the outcome is~$y[k]=v$.
 We first show that if such a~$v$ exists, it is
unique.

Clearly,  for $v< n$, we have that $p_{v+1} > p_v$. We also have
$\lim_{n \rightarrow \infty}p_n = \frac{1}{2} - \frac{1}{2e^2}$.
Thus the difference between $p_{v+1}$ and $p_v$ is:
\begin{align}
p_{v+1} - p_v &= p_v\left(1-\frac{1}{n}\right) + (1-p)\frac{1}{n}
-p_v \\
&= \frac{1-2p_v}{n} %\\
%&
> \frac{1-2p_n}{n}% \\&
>  \frac{1}{e^2n} \,.\end{align}

Hence if such a~$v$ exists, it is unique. We now show that except
with probability exponentially small in~$s$, the correct~$v$ will be
chosen. Let $X= \sum_{j=1}^{n^2 s} v[k]_j$ with~$\mu=n^2 s p_v$ the
expected value of~$X$. The participants have computed $\sigma[k] =
\frac{X}{n^2 s}$\,.

By the Chernoff bound, for any~\mbox{$0 < \delta \leq 1$},
\begin{equation}
\Pr[X \leq (1- \delta)\mu] < \exp(-\mu \delta^2/2)\,.
\end{equation}
Let $\delta = \frac{1}{2e^2np_v}$. We have
\begin{equation}
\Pr[X \leq \mu - \frac{n^2 s}{2e^2n}]  < \exp(-\frac{n^2
s}{8e^4n^2p_v})
\end{equation}
and so
\begin{equation}
\Pr[\sigma[k]_i-p_v \leq \frac{-1}{2e^2n}] <
\exp(-\frac{s}{8e^4p_v})
\end{equation}

Similarly, still by the Chernoff bound, for any~\mbox{$\delta < 2e
-1$},
\begin{equation}
\Pr[X > (1+\delta) \mu] < \exp(-\mu \delta^2/4)
\end{equation}
Let $\delta = \frac{1}{2e^2np_v}$ and we get
\begin{equation}
\Pr[X >  \mu + \frac{n^2 s}{2e^2n}] < \exp(\frac{-n^2
s}{16e^4n^2p_v})
\end{equation}
and so
\begin{equation}
\Pr[\sigma[k]_i-p_v > \frac{1}{2e^2n}] < \exp(\frac{-s}{16e^4
p_v})\,.
\end{equation}
Hence the protocol produces the correct value for~$y[k]$, except
with probability exponentially small in~$s$. \qed
\end{proof}

\begin{lemma}(Privacy)
\label{lem:privacy} In \textbf{Protocol~\ref{prot:vote1}}, no group
of corrupted voters can learn more than what they would have learned
in the ideal functionality, and this even if the protocol fails.
\end{lemma}

\begin{proof}
No assumption is made about the number of dishonest voters. The case
where all voters are corrupted is
trivially private
and in the case where only one voter is honest, his vote can be
deduced even in the ideal functionality.

When more than one voter is honest, privacy requires that, even if
the tally of the honest voters is known, the individual ballots
remain private.

In  \textbf{phase~A}, as long as at least one voter is honest, the
value of each distributed bit is perfectly hidden. In
\textbf{phase~C}, no  information is sent. %
We thus have to concentrate on \textbf{phase~B} where the voters
broadcast their information regarding each parity. Let~$H$ be the
set of honest voters. The dishonest voters learn $\bigoplus_{\ell
\in H} q_{jk\ell}$ but no information on these individual values is
revealed. The dishonest voters can thus only evaluate the
probability that this value is 1 but this information could be
deduced from the output of the ideal functionality, for instance by
fixing the corrupt participants' inputs to~$1$. \qed
\end{proof}

It is important to note that the above results do not exclude the
possibility of corrupted voters causing the protocol to fail while
still learning some information as stipulated in
Lemma~\ref{lem:privacy}. This information could unfortunately be
used to adapt the behaviour of the corrupted voters in a future
execution of~\textbf{Protocol~\ref{prot:vote1}}.

%%%%%%%%%%%%%%%%%%%%%%%%%%%%%%%%%%%%%%%%%%%%%%%%%%%%%%%%%%%%%%%%%%%%%%%%%%%%%%%%%%%%%%
\section{Voting with Authorities}
\label{sec:votewithauthorities}

In this section, we introduce a variation of the previous  voting
protocol. Our motivation is to reduce the message complexity for the
voters and reduce the need of private channels by introducing a
relatively small number of voting authorities and by only requiring
voters to communicate with these authorities. Additionally, the
simultaneous broadcast is only required among the authorities. In
this section and the following, we say that a voter constructs a
distributed bit \emph{among the authorities} if the voter creates a
distributed bit as in section~\ref{sec:vote}, except that the shares
are distributed only among the authorities. Our protocol is given as
\textbf{Protocol~\ref{prot:vote2}}.

\begin{protocol}[h]
\caption{Voting with authorities} \label{prot:vote2}

{\bf Input:} $x_i \in \{1,\ldots,m\}$
%AB New
and security parameter $s$ \\
{\bf Output:} for $k=1$ to $m$, $y[k]= \abs{\{x_j \mid x_j=k \}}$

\vspace{4pt} \hrule \vspace{4pt}

\textbf{Phase A (cast)}

For each candidate $k=1$ to~$m$, \vspace{-.25cm}
\begin{enumerate}
\item \label{step:flip-1}
Each voter~$i$ sets the value of $n^2s$ bits~$p_{ijk}$ ($j =1,
\ldots ,n^2 s$) in the following way: if~$x_i \neq k$, then all bits
are $0$; otherwise, exactly $ns$ bits (a fraction $1/n$ of the
total) are randomly chosen such that $p_{ijk}=1$ and the rest such
that $p_{ijk}=0$.

\item \label{step:parity}
For each $j =1, \ldots ,n^2 s$, each voter $i$ constructs a
distributed bit \emph{among the authorities} with value $p_{ijk}$.
Let the shares of each distributed bit be denoted $\{p_{ijk\ell}\}$
($\ell = 1, \ldots r$)
\end{enumerate}

\textbf{Phase B (broadcast)}

All authorities $\ell$, for every $j$ and $k$ simultaneously
broadcast $q_{jk\ell}=\bigoplus_i p_{ijk\ell}$

\textbf{Phase C (tally)}

To compute the tally, $y[k]$, for each value $k=1, \ldots ,m$, each
participant sets: \mbox{$v[k]_j=\bigoplus_{\ell=1}^n q_{jk\ell}$},
\mbox{$\sigma[k] = \sum_{j=1}^{n^2 s} \frac{v[k]_j}{n^2s}$} and  if
there exists an integer~$v$ such that $\left|\sigma[k] - p_v \right|
< \frac{1}{2e^2n} \,$, where $p_v =
\frac{1}{2}\left(\frac{n-2}{n}\right)^v
\left(\left(\frac{n}{n-2}\right)^v -1\right)$, then $y[k]=v$\,.

If for any~$m$, no such value~$v$ exists, or if $\sum_{k=1}^m y[k]
\neq n$, the protocol fails.

Each authority broadcasts the outcome of the tally,
if there is any disagreement, the protocol fails.

\end{protocol}

The complexity of \textbf{Protocol~\ref{prot:vote2}} is as follows:
$n$ voters each create $mn^2s$ distributed bits, which are
distributed among~$r$ authorities,  for a total of $nr$ messages of
size $mn^2s$. \textbf{Phase B} requires a single simultaneous
broadcast among $r$ authorities, each sending a message of size~$m
n^2 s$. \textbf{Phase C} requires  $r$ broadcasts of size as most $m
\log n$.

\begin{lemma}(Correctness)
\label{lem2:correctness corrupt}
If at least one authority is
honest, and if \textbf{Protocol~\ref{prot:vote2}} does not fail, the
result of the vote is consistent with the vote of the honest voters
and some non-adaptive choice for the dishonest voters, except with
probability exponentially small in~$s$.
\end{lemma}

\begin{proof}
The proof is obtained by replacing voters by authorities at the
appropriate place in proof of Lemma~\ref{lem:correctness corrupt}.
It is important here that the correctness probability only depends
on~$s$ and not on the number of voters or authorities. \qed
\end{proof}

\begin{lemma}(Privacy)
\label{lem2:privacy} In \textbf{Protocol~\ref{prot:vote2}}, if at
least one authority is honest, no collusion of dishonest voters and
authorities can learn more than what they would have learned in the
ideal functionality, and this even if the protocol fails.
\end{lemma}

\begin{proof}
The proof is very  similar to the proof of Lemma~\ref{lem:privacy}.
In \textbf{Protocol~\ref{prot:vote2}}, part of the work performed by
the voters in \textbf{Protocol~\ref{prot:vote1}} is done by the
authorities. If at least one authority is honest, there is no way
dishonest participants (voters or authorities) can learn any
information about the value of the distributed bit created by an
honest voter. The rest of the argument is the same as in
Lemma~\ref{lem:privacy}.\qed
\end{proof}

Note that in \textbf{Protocol~\ref{prot:vote2}}, any participant can
make the protocol fail. Voters can do this, for instance, by setting
an abnormally high number of distributed bits to~1, and authorities
can do this by changing their inputs into the simultaneous
broadcast. Furthermore, note  that in \textbf{Phase~B}, although the
simultaneous broadcast happens among the authorities, it is not a
problem if the voters are passive listeners. At the end of
\textbf{Phase~C}, the authorities broadcast the result of the tally.
We required unanimity of these messages in order to declare that the
protocol has succeeded.

%%%%%%%%%%%%%%%%%%%%%%%%%%%%%%%%%%%%%%%%%%%%%%%%%%%%%%%%%%%%%%%%%%%%%%%%%%%%%%%%%%%%%%
\section{Voting with Authorities and Verification}
\label{sec:votewithverification}

One of the issues with the previous two protocols is that any voter
can cause them to fail by introducing noise. In this section, we use
the cut-and-choose technique, augmented with an equality test,  to
allow authorities to revoke a noisy ballot. This is done by having
each voter distribute many encrypted but identical \emph{votes},
where a \emph{vote} is $k$ lists of $n^2s$ bits (as created, for
instance, in step~1 of  \textbf{Phase A} of
\textbf{Protocol~\ref{prot:vote2}}). A vote is \emph{correct} if its
contents correspond to the construction of step~1 of \textbf{Phase
A} of \textbf{Protocol~\ref{prot:vote2}}, i.e.~all bits are even
except one candidate which has exactly $ns$ bits sets to~1. The
authorities then open half of the votes and verify the correctness;
a subsequent step will ensure that the unopened votes are equal,
thus providing exponential security.

Our protocol is presented as \textbf{Protocol~\ref{prot:vote3}}, in
which the authorities use the following two  simple routines.

\emph{Random choices:} authorities can generate common random bits
in the following way. Each authority locally generates a random bit,
after which all authorities simultaneously broadcast these bits. The
common random bit is set to be the parity of the broadcast bits.
Obviously, this value is truly random if at least one authority is
honest. This process can be done in parallel, requiring only one
simultaneous broadcast.

\emph{Distributed bit equality:} suppose the authorities share two
distributed bits. They can verify if these two distributed bits have
the same value without revealing this value. Let
$a=\bigoplus_{i=1}^r a_i$ and $b= \bigoplus_{i=1}^r b_i$ be the two
distributed bits. Each authority~$i$ simultaneously broadcasts
$c_i=a_i \oplus b_i$. If~$\bigoplus_{i=1}^r c_i=0$ then the
distributed bits are equal (unless an authority is cheating). A
dishonest authority can make the protocol output the wrong answer,
but under no circumstance will  this process reveal any information
about the values of $a$ or~$b$.

\begin{protocol}[h]
\caption{Voting with authorities and verification}
\label{prot:vote3}

{\bf Input:} $x_i \in \{1,\ldots,m\}$
%AB New
and security parameter $s$  \\
{\bf Output:} for $k=1$ to $m$, $y[k]= \abs{\{x_j \mid x_j=k \}}$ as
well as a list of voters with revoked ballots
 \vspace{4pt} \hrule \vspace{4pt}

\textbf{Phase A (randomness)}

The authorities generate enough common random bits.

\textbf{Phase B (verification and vote casting)}

For each voter: \vspace{-.25cm}
\begin{enumerate}
\item Each voter executes step~1 of \textbf{Phase A} of
\textbf{Protocol~\ref{prot:vote2}}, thus creating one \emph{vote}.

\item $2s$ copies of the vote are made, and for each vote,  the shares of the distributed
bits are computed as in step~2 of \textbf{Phase A} of
\textbf{Protocol~\ref{prot:vote2}} (the shares are independently
randomly chosen).
\item  Each vote is encrypted with two random
permutations: the first permutation changes the order of the $k$
candidates, and the second permutation changes the order of the
$n^2s$ distributed bits (the same permutation is applied for each
candidate within a vote).
\item \label{prot3:distribute-votes}The shares of the encrypted votes are distributed among the
authorities.

\item \label{prot3:simbroadcast} The authorities randomly choose $s$ votes and
simultaneously broadcast all bits involved in these votes.

\item If any of the opened votes is not correct, the voter's ballot is revoked.

\item \label{prot3:reveal}Each authority reveals to the voter which votes were opened.
If the voter receives  inconsistent messages, his ballot is revoked.

\item \label{prot3:voter-reveal} For the $s$ remaining votes, the voter reveals to the authorities  both the
permutation that was applied on the distributed bits and the
permutation that was applied on the candidates. The authorities
permute their shares of the remaining votes so that all votes are
equal.

\item \label{prot3:dist-equality}The authorities perform \emph{distributed bit equality} tests between each
distributed bit of the first remaining vote and all corresponding
distributed bits for all other remaining votes. If any of these
tests fail, then the voter's ballot is revoked. If all tests
succeed, all but the first remaining vote are discarded.
\end{enumerate}
 \vspace{-.25cm}
\textbf{Phase C (broadcast and tally)}

\textbf{Phases B} and \textbf{C} of
\textbf{Protocol~\ref{prot:vote2}} are performed with  all remaining
non-revoked votes.

\end{protocol}

Note that in \textbf{Protocol~\ref{prot:vote3}}, any dishonest
authority can make the protocol fail and any authority can
dishonestly revoke any voter's ballot.

The complexity of \textbf{Protocol~\ref{prot:vote3}} is as follows:
each of the $n$ voters sends $r$~messages of size  $2 m n^2 s^2$ for
the votes (step~\ref{prot3:distribute-votes}) and $r$~messages of
size $n^2 s^2 \log(n^2 s)+sm\log(m)$ for the permutations
(step~\ref{prot3:voter-reveal}). In order to generate enough random
bits, the authorities are involved in a single simultaneous
broadcast of size $n \log(\binom{2s}{s}) \in O(ns)$. For the rest of
the protocol, the $r$ authorities are involved in
step~\ref{prot3:simbroadcast} in a simultaneous broadcasts of size
$m n^2 s^2$ for each voter; in step~\ref{prot3:reveal}, they require
a message of size~$s$ for each voter, and in
step~\ref{prot3:dist-equality}, they broadcast $(s-1)mn^2s$ bits.
\textbf{Phase C} requires one last simultaneous broadcast of size $m
n^2 s$ as well as $r$ broadcasts of size as most $m \log n$.

\begin{lemma}(Correctness) \label{lem3:correctness corrupt}
If at least one authority is
honest, and if \textbf{Protocol~\ref{prot:vote3}} does not fail,
then every ballot that is not revoked is correctly counted except
with probability exponentially small in~$s$.
\end{lemma}

\begin{proof}
The proof is  identical to the proof of Lemma~\ref{lem2:correctness
corrupt}.
The verification of the vote only makes the protocol more
robust. \qed
\end{proof}

\begin{lemma}(Privacy)
\label{lem3:privacy} In \textbf{Protocol~\ref{prot:vote3}}, if at
least one authority is honest, no collusion of dishonest voters and
authorities can learn more than what they would have learned in the
ideal functionality.
\end{lemma}

\begin{proof}
To see that privacy of the vote is guaranteed if at least one
authority is honest, we first observe that \textbf{phase B} of the
protocol does not reveal information about the voters' choice; it
only ensures correctness of the vote. Once this phase is done, the
rest of the protocol is identical to
\textbf{Protocol~\ref{prot:vote2}} and the same argument as in
Lemma~\ref{lem2:privacy} can be used here. \qed
\end{proof}

As mentioned at the beginning of this section, in
\textbf{Protocols~\ref{prot:vote1}} and \textbf{\ref{prot:vote2}}, a
voter can vote in an inconsistent way, causing the protocol to fail
with very high probability. In \textbf{Protocol~\ref{prot:vote3}}
the votes are verified: if a vote is not correct, there is only a
probability exponentially small in~$s$ that the vote will not be
revoked. Thus, dishonest voters can only make the protocol fail with
exponentially small probability in $s$. We formalize this below.

\begin{lemma}(Robustness)
\label{lem3:robustness} No coalition of voters can alone make the
protocol fail, except with exponentially small probability in~$s$.
\end{lemma}

\begin{proof}
The only way for a voter not to provide the correct information in
\textbf{phase~B} is to generate incorrect votes. Since half of the
votes are opened, and the other half is checked for equality, the
only way for a voter to successfully provide an incorrect ballot is
for the $s$ opened votes to be correct and the $s$ remaining votes
to be incorrect, yet identical. This happens with exponentially
small probability in $s$. \qed
\end{proof}

\section{Conclusion}

We presented three voting scheme with unconditional security and
information-theoretic correctness,  without assuming any bound on
the number of corrupt voters or voting authorities. For this to
succeed, we had to assume pairwise private channels and a
simultaneous broadcast channel (as discussed, this assumption can be
replaced by temporary computational assumptions, yielding
everlasting security). We also had to allow any participant to cause
the protocol to fail. Fortunately, we were able to relax some of the
above assumptions in \textbf{Protocols~\ref{prot:vote1}}
and~\textbf{\ref{prot:vote3}} by introducing a set of voting
authorities.

We are currently considering a tradeoff between the revoking power
of authorities and the correctness of the protocol. This can be
achieved as a modification of \textbf{Protocol~\ref{prot:vote3}} by
randomly grouping the authorities and by performing the protocol in
parallel within each group.

Although our initial motivation was of theoretical nature, we
believe that this work might lead to interesting practical
applications.

\section*{Acknowledgements}
%AB new
The authors wish to thank S\'ebastien Gambs for proofreading and
Jeroen van de Graaf for suggesting that  we write up and submit our
ideas.
%AB any funding?

\bibliographystyle{alpha}
\bibliography{references}

\newcommand{\noopsort}[1]{}
\begin{thebibliography}{BOGW88}

\bibitem[BOGW88]{BGW88}
M.~Ben-Or, S.~Goldwasser, and A.~Wigderson.
\newblock Completeness theorems for non-cryptographic fault-tolerant
  distributed computation.
\newblock In {\em Proceedings of the 20th annual ACM Symposium on Theory of
  Computing (STOC)}, pages 1--10, 1988.

\bibitem[BT07]{BroadbentTapp}
A.~Broadbent and A.~Tapp.
\newblock Information-theoretic security without an honest majority.
\newblock In {\em Proceedings of the 13th International Conference on the
  Theory and Application of Cryptology and Information Security (ASIACRYPT
  '07)}, pages 410--426, 2007.

\bibitem[CCD88]{CCD88}
D.~Chaum, C.~Cr{\'e}peau, and I.~Damg{\r a}rd.
\newblock Multiparty unconditionally secure protocols.
\newblock In {\em Proceedings of the 20th annual ACM Symposium on Theory of
  Computing (STOC)}, pages 11--19, 1988.

\bibitem[CGMA85]{CGMW85}
B.~Chor, S.~Goldwasser, S.~Micali, and B.~Awerbuch.
\newblock Verifiable secret sharing and achieving simultaneity in the presence
  of faults.
\newblock In {\em Proceedings of the 26th annual IEEE Symposium on Foundations
  of Computer Science (FOCS)}, pages 383--395, 1985.

\bibitem[Cha88]{Chaum88}
D.~Chaum.
\newblock The dining cryptographers problem: Unconditional sender and recipient
  untraceability.
\newblock {\em Journal of Cryptology}, 1:65--75, 1988.

\bibitem[Gra08]{Jeroen}
J.~{\noopsort{Graaf}}~{van de} Graaf.
\newblock Private Communication, 2008.

\bibitem[HM05]{HD05}
A.~Hevia and D.~Micciancio.
\newblock Simultaneous broadcast revisited.
\newblock In {\em Proceedings of the 24th annual ACM symposium on Principles of
  distributed computing}, pages 324--333, 2005.

\bibitem[Ken99]{Kent99}
A.~Kent.
\newblock Unconditionally secure bit commitment.
\newblock {\em Physical Review Letters}, 83:1447--1450, 1999.

\bibitem[RBO89]{RB89}
T.~Rabin and M.~Ben-Or.
\newblock Verifiable secret sharing and multiparty protocols with honest
  majority.
\newblock In {\em Proceedings of the 21st annual ACM Symposium on Theory of
  Computing (STOC)}, pages 73--85, 1989.

\end{thebibliography}

\end{document}